\documentclass[11pt]{llncs}

\usepackage[english]{babel}
\usepackage[T1]{fontenc}
\usepackage[latin1]{inputenc}
\usepackage{times}
\usepackage{amssymb,amsmath,amscd,latexsym,graphicx}
\usepackage{tikz}
\usepackage{gastex}
\usepackage{lineno}

\newcommand{\set}[1]{\{ #1 \}}
\newcommand{\seq}[1]{\langle #1 \rangle}

\newcommand{\Sig}[1]{\Sigma_{#1}}
\newcommand{\Sigstar}{\Sigma^*}
\newcommand{\fin}{\mathsf{UniCloOmg}}
\newcommand{\PSPACE}{\mathrm{PSPACE}}
\newcommand{\NPSPACE}{\mathrm{NPSPACE}}
\newcommand{\NLOGSPACE}{\mathrm{NLOGSPACE}}

\newcommand{\A}{\mathcal{A}}
\newcommand{\B}{\mathcal{B}}
\newcommand{\C}{\mathcal{C}}

\newcommand{\La}{\mathcal{L}}
\newcommand{\LA}{\mathcal{L(A)}}
\newcommand{\LB}{\mathcal{L(B)}}
\newcommand{\N}{\mathbb{N}}
\newcommand{\infi}{\mathrm{Inf}}
\newcommand{\So}{\Sigma^{\omega}}
\newcommand{\streett}{\mathrm{Streett}(R,G)}
\newcommand{\distk}{\mathrm{dist}_k}
\newcommand{\nextk}{\mathrm{next}_k}
\newcommand{\omegareg}{\mathbb{L}_{\omega}}
\newcommand{\finp}{\mathrm{FinParity}(p)}
\newcommand{\fstreettrg}{\mathrm{FinStreett}(R,G)}
\newcommand{\fbucf}{\mathrm{FinB\ddot{u}chi}(F)}
\newcommand{\parp}{\mathrm{Parity}(p)}

\newcommand{\bucf}{\mathrm{B\ddot{u}chi}(F)}
\newcommand{\pref}[1]{\mathrm{pref}(#1)}

\newcommand{\Final}{\mathrm{Final}}

\newcommand{\DB}{\mathit{DB}}
\newcommand{\DP}{\mathit{DP}}
\newcommand{\DS}{\mathit{DS}}

\newcommand{\DFB}{\mathit{DFB}}
\newcommand{\DFP}{\mathit{DFP}}
\newcommand{\DFS}{\mathit{DFS}}

\newcommand{\NB}{\mathit{NB}}
\newcommand{\NP}{\mathit{NP}}
\newcommand{\NS}{\mathit{NS}}

\newcommand{\NFB}{\mathit{NFB}}
\newcommand{\NFP}{\mathit{NFP}}
\newcommand{\NFS}{\mathit{NFS}}
\newcommand{\Acc}{\mathit{Acc}}

\newtheorem{notation}{Notation}

\def\qed{\rule{0.4em}{1.4ex}}

\begin{document}
\title{Finitary Languages}
\author{
Krishnendu Chatterjee \inst{1} \and
Nathana\"el Fijalkow \inst{1,2}  
}
\institute{IST Austria (Institute of Science and Technology, Austria) 
\email{krishnendu.chatterjee@ist.ac.at}
\and
\'ENS Cachan (\'Ecole Normale Supérieure de Cachan, France)
\email{nathanael.fijalkow@gmail.com}
}

\maketitle

\begin{abstract}
The class of $\omega$-regular languages provides a robust specification 
language in verification.
Every $\omega$-regular condition  can be decomposed into a safety part 
and a liveness part.
The liveness part ensures that something good happens ``eventually''.
Finitary liveness was proposed by Alur and Henzinger as a stronger formulation of liveness~\cite{AH98-TOPLAS}.
It requires that there exists an unknown, fixed bound $b$ 
such that something good happens within $b$ transitions.
In this work we consider automata with finitary acceptance conditions defined by finitary B\"uchi, parity and Streett languages.
We give their topological complexity of acceptance conditions, and present a regular-expression characterization of the languages they express.
We provide a classification of finitary and classical automata with respect to the expressive power, and give optimal algorithms for classical decisions questions on finitary automata.
We (a)~show that the finitary languages are $\Sig{2}$-complete;
(b)~present a complete picture of the expressive power of various 
classes of automata with finitary and infinitary acceptance conditions;
(c)~show that the languages defined by finitary parity 
automata exactly characterize the star-free fragment of $\omega B$-regular 
languages; and 
(d)~show that emptiness is $\NLOGSPACE$-complete and universality as well as
language inclusion are $\PSPACE$-complete for finitary automata.
\end{abstract}

\section{Introduction}

\noindent{\bf Classical $\omega$-regular languages: strengths and weakness.} 
The class of $\omega$-regular languages provides a robust language for 
specification for solving control and verification problems
(see, \textit{e.g}, \cite{PR89-POPL,RW87-SJCO}).
Every $\omega$-regular specification 
can be decomposed into a safety part and a liveness part~\cite{AS85-IPL}.
The safety part ensures that the component will not do anything ``bad''
(such as violate an invariant) within any finite number of transitions.
The liveness part ensures that the component will do something ``good''
(such as proceed, or respond, or terminate) in the long-run.
Liveness can be violated only in the limit, by infinite sequences of 
transitions, as no bound is stipulated on when the ``good'' thing must 
happen.
This infinitary, classical formulation of liveness has both strengths and 
weaknesses. 
A main strength is robustness, in particular, independence from the chosen 
granularity of transitions.
Another main strength is simplicity, allowing liveness to serve as an 
abstraction for complicated safety conditions.
For example, a component may always respond in a number of transitions 
that depends, in some complicated manner, on the exact size of the 
stimulus.
Yet for correctness, we may be interested only that the component will 
respond ``eventually''.
However, these  strengths also point to a weakness of the classical 
definition of liveness:
it can be satisfied by components that in practice are quite
unsatisfactory because no bound can be put on their response time.

\smallskip\noindent{\bf Stronger notion of liveness.}
For the weakness of the infinitary formulation of liveness, alternative and 
stronger formulations of liveness have been proposed.
One of these is {\em finitary} liveness \cite{AH98-TOPLAS}:
finitary liveness does not insist on a response within a known bound~$b$
(\textit{i.e}, every stimulus is followed by a response within $b$ transitions), 
but on response within some unknown bound
(\textit{i.e}, there exists $b$ such that every stimulus is followed by a response 
within $b$ transitions).  
Note that in the finitary case, the bound $b$ may be arbitrarily large, 
but the response time must not grow forever from one stimulus to the next.
In this way, finitary liveness still maintains the robustness (independence 
of step granularity) and simplicity (abstraction of complicated safety) of 
traditional liveness, while removing unsatisfactory implementations.

\smallskip\noindent{\bf Finitary parity and Streett conditions.}
The classical infinitary notion of fairness is given by the Streett condition:
it consists of a set of $d$ pairs of requests and 
corresponding responses (grants) and requires that every request 
that appears infinitely often must be responded infinitely often. 
Its finitary counterpart, the finitary Streett condition requires that there is a bound $b$ such that in the limit every request is responded within $b$ steps.
The classical infinitary parity condition consists of a priority function and requires that the minimum priority visited infinitely often is even.
Its finitary counterpart, the finitary parity condition requires that there is a bound $b$ such that in the limit after every odd priority a lower even priority is visited within $b$ steps.

\smallskip\noindent{\bf Results on classical automata.}
There are several robust results on the languages expressible by automata with infinitary B\"uchi, parity and Streett conditions, as follows:
(a)~\emph{Topological complexity:} it is known that B\"uchi languages are $\Pi_2$-complete, whereas parity and Streett languages lie in the boolean closure of $\Sig{2}$ and $\Pi_2$~\cite{MP92};
(b)~\emph{Automata expressive power:} non-deterministic automata
with B\"uchi conditions have the same expressive power as 
deterministic and non-deterministic parity and Streett 
automata~\cite{Cho74-JCSS,Sa92-STOC};
and (c)~\emph{Regular-expression characterization:} the class of languages expressed by deterministic parity is exactly defined by $\omega$-regular expressions (see the handbook~\cite{Tho97} for details).

\smallskip\noindent{\bf Our results.} For finitary B\"uchi, parity and Streett
languages, topological, automata-theoretic, regular-expression and decision problems studies were all missing.
In this work we present results in the four directions, as follows:
\begin{enumerate}
\item \emph{Topological complexity.} We show that finitary B\"uchi, parity and Streett conditions are $\Sig{2}$-complete.
\item \emph{Automata expressive power.} 
We show that finitary automata are incomparable in expressive power with classical automata.
As in the infinitray setting, we show that non-deterministic automata with finitary B\"uchi, parity and Streett conditions have the same expressive power, as well as deterministic parity and Streett automata, which are strictly more expressive than 
deterministic finitary B\"uchi automata.
However, in contrast to the infinitary case, for finitary parity
condition, non-deterministic automata are strictly more expressive than
the deterministic counterpart.
As a by-product we derive boolean closure properties for finitary automata.
\item \emph{Regular-expression characterization.} We consider the characterization of finitary automata through an extension of $\omega$-regular languages defined as $\omega B$-regular languages by~\cite{BC06-LICS}.
We show that languages defined by non-deterministic finitary B\"uchi automata
are exactly the star-free fragment of $\omega B$-regular languages. 
\item \emph{Decision problems.} We show that emptiness is $\NLOGSPACE$-complete and universality as well as language inclusion are $\PSPACE$-complete for finitary automata.
\end{enumerate}

\noindent{\bf Related works.}
The notion of finitary liveness was proposed and studied in~\cite{AH98-TOPLAS},
and games with finitary objectives was studied in~\cite{CHH09-ToCL}. 
A generalization of $\omega$-regular languages as $\omega B$-regular languages 
was introduced in~\cite{BC06-LICS} and variants have been studied in~\cite{BT09-FSTTCS} (also see~\cite{Bojanczyk10-STACS} for a survey);
a topological characterization has been given in~\cite{HST-MFCS10}.
Our work along with topological and automata-theoretic studies of finitary languages, explores the relation between finitary languages and $\omega B$-regular
expressions, rather than identifying a subclass of $\omega B$-regular expressions. 
We identify the exact subclass of $\omega B$-regular expressions that corresponds
to non-deterministic finitary parity automata.

\section{Definitions}


\subsection{Languages topological complexity}

Let $\Sigma$ be a finite set, called the alphabet.
A word $w$ is a sequence of letters, which can be either finite or 
infinite.
A language is a set of words: $L \subseteq \Sigstar$ is a 
language over finite words and $L \subseteq \So$ over infinite words.

\smallskip\noindent{\bf Cantor topology and Borel hierarchy.} 
Cantor topology on $\So$ is given by \emph{open} sets:
a language is open if it can be described as $W \cdot \So$ where $W \subseteq \Sigstar$.
Let $\Sig{1}$ denote the open sets and $\Pi_1$ denote the closed sets (a language is closed if its complement is open): they form the first level of the Borel hierarchy. 
Inductively, we define:
$\Sig{i+1}$ is obtained as countable union of $\Pi_i$ sets; and
$\Pi_{i+1}$ is obtained as countable intersection of $\Sig{i}$ sets.
The higher a language is in the Borel hierarchy, the higher its topological complexity.

Since the above classes are closed under continuous preimage, we can define the notion of Wadge reduction~\cite{Wadge-Thesis}: $L$ reduces to $L'$, denoted by $L \preceq L'$, if there exists a continuous function $f : \So \rightarrow \So$ such $L = f^-(L')$, where $f^-(L')$ is the preimage of $L'$ by $f$.
A language is hard with respect to a class if all languages of this class 
reduce to it. If it additionally belongs to this class, then it is complete.

For $L \subseteq \So$, let $\pref{L} \subseteq \Sigstar$ be the set of finite 
prefixes of words in $L$. 
The following property holds:

\begin{proposition}
For all languages $L \subseteq \So$, $L$ is closed if and only if,
for all infinite words $w$, if all finite prefixes of $w$ are in $\pref{L}$, then $w \in L$.
\end{proposition}

\smallskip\noindent{\bf Classical liveness conditions.} 
We now consider three classes of languages that are widespread in verification and specification. They define liveness properties, \textit{i.e}, intuitively say that something good will happen ``eventually''.
For an infinite word $w$, let $\infi(w) \subseteq \Sigma$ denote the set of letters that appear infinitely often in $w$. 
The class of B\"uchi languages is defined as follows, given $F \subseteq \Sigma$:
$$\bucf  =  \set{ w \mid \infi(w) \cap F \not= \emptyset}$$
\textit{i.e}, the B\"uchi condition requires that some letter in $F$ appears infinitely often.
The class of parity languages is defined as follows, given $p : \Sigma \rightarrow \N$ a priority function that maps letters to integers (representing priorities):
$$\parp = \set{w \mid \min(p(\infi(w))) \mbox{ is even}}$$
\textit{i.e}, the parity condition requires that
the lowest priority the appears infinitely often is even.
The class of Streett languages is defined as follows, given $(R,G) = (R_i,G_i)_{1 \leq i \leq d}$, where $R_i,G_i \subseteq \Sigma$ are request-grant pairs:
$$\streett = \set{w \mid \forall i, 1 \leq i \leq d, \infi(w) \cap R_i \neq \emptyset \Rightarrow \infi(w) \cap G_i \neq \emptyset}$$
\textit{i.e}, the Streett condition requires that for all requests $R_i$, if it appears infinitely often, then the corresponding  grant $G_i$ also appears infinitely often.

The following theorem presents the topological complexity of the classical languages:

\begin{theorem}[Topological complexity of classical languages~\cite{MP92}]\label{thrm_top}
\begin{itemize}
\vspace{-2.5mm}
	\item For all $\emptyset \subset F \subset \Sigma$, the language
$\bucf$ is $\Pi_2$-complete.
	\item The parity and Streett languages lie in the boolean closure of 
$\Sig{2}$ and $\Pi_2$.
\end{itemize}
\end{theorem}

\subsection{Finitary languages}

The finitary parity and Streett languages have been defined in~\cite{CHH09-ToCL}. We recall their definitions, and also specialize them to finitary B\"uchi languages.
Let $(R,G) = (R_i,G_i)_{1 \leq i \leq d}$, where $R_i,G_i \subseteq \Sigma$, the definition for $\fstreettrg$ uses distance sequence as follows:
$$\distk^j (w,(R,G)) =
\begin{cases}
0 & w_k \notin R_j\\
\inf \{k'-k \mid k' \geq k, w_{k'} \in G_j \} & w_k \in R_j 
\end{cases}
$$
\textit{i.e}, given a position $k$ where $R_j$ is requested, $\distk^j (w,(R,G))$ is the time steps (number of transitions) between the request $R_j$ and the corresponding grant $G_j$.
Note that $\inf(\emptyset) = \infty$.
Then $\distk (w,(R,G)) = \max\set{\distk^j(w,p) \mid 1 \leq j \leq d}$ and:
$$\fstreettrg = \set{w \mid \limsup_k \distk (w,(R,G)) < \infty}$$
\textit{i.e}, the finitary Streett condition requires the supremum limit of the distance sequence to be bounded.

Since parity languages can be considered as a particular case of Streett languages, where $G_1 \subseteq R_1 \subseteq G_2 \subseteq R_2 \ldots$, the latter allows to define $\finp$. 
The same applies to finitary B\"uchi languages, which is a particular case of finitary parity languages where the letters from the set $F$ have priority $0$ and others have priority $1$. 
We get the following definitions.
Let $p : \Sigma \rightarrow \N$ a priority function, we define:
$$\distk (w,p) = \inf \set{k'-k \mid k' \geq k, p(w_{k'}) \textrm{ is even and } p(w_{k'}) \leq p(w_k)}$$
\textit{i.e}, given a position $k$ where $p(w_k)$ is odd, $\distk (w,p)$ is the time steps between the odd priority $p(w_k)$ and a lower even priority.
Then $\finp = \set{ w \mid \limsup_k \distk (w,p) < \infty }$.
We define similarly the finitary B\"uchi language: given $F \subseteq \Sigma$, let:
$$\nextk (w,F) = \inf \{k'-k \mid k' \geq k, w_{k'} \in F \}$$
\textit{i.e}, $\nextk (w,F)$ is the time steps before visiting a letter in $F$.
Then $\fbucf = \set{w \mid \limsup_k \nextk (w,F) < \infty}$.

\subsection{Automata, $\omega$-regular and finitary languages}

\begin{definition} An automaton is a tuple $\A = (Q, \Sigma, Q_0, \delta, \Acc)$, where $Q$ is a finite set of states, 
$\Sigma$ is the finite input alphabet, 
$Q_0 \subseteq Q$ is the set of initial states, 
$\delta \subseteq Q \times \Sigma \times Q$ is the transition relation and 
$\Acc \subseteq Q^{\omega}$ is the acceptance condition.
\end{definition}

An automaton is deterministic if it has a single initial state and 
for every state and letter there is at most one transition.
The transition relation of deterministic automata are described by functions $\delta : Q \times \Sigma \rightarrow Q$.
An automaton is complete if for every state and letter there is a transition.
This is the case when the transition function is \emph{total}.

\medskip\noindent{\bf Runs.} 
A run $\rho = q_0 q_1 \dots$ is a word over $Q$, where $q_0 \in Q_0$. 
The run $\rho$ is accepting if it is infinite and $\rho \in \Acc$.
We will write $p \xrightarrow{a} q$ to denote $(p,a,q) \in \delta$.
An infinite word $w = w_0 w_1 \ldots$ induces possibly several runs of $\A$:
a word $w$ induces a run $\rho = q_0 q_1 \ldots$ if for all $n \in \N, q_n \xrightarrow{w_n} q_{n+1} \dots$.
The language accepted by $\A$, denoted by $\LA \subseteq \Sigma^\omega$, is:
$$\LA = \set{ w \mid \mbox{there exists an accepting run } \rho \mbox{ induced by } w}.$$

\medskip\noindent{\bf Acceptance conditions.} 
We will consider various acceptance conditions for automata obtained
from the last section by considering $Q$ as the alphabet.
For example, given $F \subseteq Q$, the languages $\bucf$ and $\fbucf$ define B\"uchi and finitary B\"uchi acceptance conditions, respectively.
Automata with finitary acceptance conditions are referred as finitary automata; classical automata are those equipped with infinitary acceptance conditions.

\begin{notation}
We use a standard notation to denote the set of languages recognized by some 
class of automata. The first letter is either $N$ or $D$, where $N$ stands 
for ``non-deterministic'' and $D$ stands for ``deterministic''. 
The last letter refers to the acceptance condition:
$B$ stands for ``B\"uchi'', $P$ stands for ``parity'' 
and $S$ stands for ``Streett''. The acceptance condition may be prefixed by 
$F$ for ``finitary''. For example, $\NP$ denotes non-deterministic parity
automata, and $\DFS$ denotes deterministic finitary Streett automata. 
We have the following combination:
$$\left \{ \begin{array}{c} N \\ D \end{array} \right \} \cdot
\left \{ \begin{array}{c} F \\ \varepsilon \end{array} \right \} \cdot
\left \{ \begin{array}{c} B \\ P \\ S \end{array} \right \}$$
\end{notation}

\noindent We denote by $\omegareg$ the class of languages accepted by deterministic parity automata.
The following theorem summarizes the results of expressive power of classical automata~\cite{Buc62-CLMPS,Sa92-STOC,Cho74-JCSS,GH82-STOC}:

\begin{theorem}[Expressive power results for classical automata]\label{classical_aut}
$$\DB \subset \omegareg \doteq \NB = \DP = \NP = \DS = \NS$$
\end{theorem}

\section{Topological complexity}

In this section we define a finitary operator $\fin$ that allows us to describe finitary B\"uchi, finitary parity and finitary Streett languages topologically and to relate them to the classical B\"uchi, parity and Streett languages;
we then give their topological complexity.

\medskip\noindent{\bf Union-closed-omega-regular operator on languages.}
Given a language $L \subseteq \So$, the language $\fin(L) \subseteq \So$ is 
the \emph{union} of the languages $M$ that are subsets of $L$,
\emph{$\omega$-regular} and \emph{closed}, \textit{i.e},
$\fin(L) = \bigcup \set{M \mid M \subseteq L, M \in \Pi_1, M \in \omegareg}$.

\begin{proposition}\label{fin} 
For all languages $L \subseteq \So$ we have $\fin(L) \in \Sig{2}$.
\end{proposition}

\begin{proof} 
Since the set of finite automata can be enumerated in sequence, it follows
that $\omegareg$ is countable.
So for all languages $L$, the set $\fin(L)$ is described as a countable
union of closed sets.
Hence $\fin(L) \in \Sig{2}$.
\hfill\qed
\end{proof}

We present a \emph{pumping lemma} for $\omega$-regular languages that we will use to prove the topological complexity of finitary languages.

\begin{lemma}[A pumping lemma]\label{lemm_pump1}
Let $M$ be an $\omega$-regular language.
There exists $n_0$ such that for all words $w \in M$, 
for all positions $j \geq n_0$, there exist $j \leq i_1 < i_2 \leq j + n_0$
such that for all $\ell \geq 0$ we 
have 
$w_0 w_1 w_2 \ldots w_{i_1-1} \cdot (w_{i_1} w_{i_1+1} \ldots w_{i_2-1})^\ell \cdot 
w_{i_2} w_{i_2 +1} \ldots 
\in M$.
\end{lemma}

\begin{proof}
Given $M$ is a $\omega$-regular language, let $\A$ be a complete and deterministic parity automata that recognizes $M$,
and let $n_0$ be the number of states of $\A$.
Consider a word $w = w_0 w_1 w_2 \ldots$ such that $w \in M$, and let 
$\rho=q_0 q_1 q_2 \ldots$ be the unique run induced by $w$ in $\A$. 
Consider a position $j$ in $w$ such that $j \geq n_0$.
Then there exist $j \leq i_1 < i_2 \leq j + n_0$ such that
$q_{i_1} = q_{i_2}$, this must happen as $\A$ has $n_0$ states.
For $\ell \geq 0$, if we consider the word
$w^\ell= 
w_0 w_1 w_2 \ldots w_{i_1-1} \cdot (w_{i_1} w_{i_1+1} \ldots w_{i_2-1})^\ell \cdot 
w_{i_2} w_{i_2 +1} \ldots$, then the unique run induced by 
$w^\ell$ in $\A$ is 
$\rho^\ell=q_0 q_1 q_2 \ldots q_{i_1-1} \cdot (q_{i_1} q_{i_1+1} \ldots q_{i_2-1})^\ell \cdot 
q_{i_2} q_{i_2 +1} \ldots$.
Since the parity condition is independent of finite prefixes and the run $\rho$ 
is accepted by $\A$, it follows that $\rho^\ell$ is accepted by $\A$.
Since $\A$ recognizes $M$, we have $w^\ell \in M$.
\hfill \qed
\end{proof}

The following lemma shows that $\fstreettrg$ is obtained by applying the $\fin$ operator to $\streett$.

\begin{lemma}\label{fins}
For all $(R,G) = (R_i,G_i)_{1 \leq i \leq d}$, where $R_i,G_i \subseteq \Sigma$, we have
$$\fin(\streett) = \fstreettrg.$$
\end{lemma}

\begin{proof} We present the two directions of the proof.
\begin{enumerate}
	\item We first show that $\fin(\streett) \subseteq \fstreettrg$.
Let $M \subseteq \streett$ such that $M$ is closed and $\omega$-regular. 
Let $w = w_0 w_1 \ldots \in M$, and assume towards contradiction, 
that $\limsup_k \distk (w,(R,G)) = \infty$. Hence for all $n_0 \in \N$, 
there exists $n \in \N$ such that $n \geq n_0$ and 
$\mathrm{dist}_n (w,(R,G)) \geq n_0$.
Let $n_0 \in \N$ given by the pumping lemma on $M$, 
from above given $n_0$ we obtain $j$ such that $j \geq n_0$ and 
$\mathrm{dist}_j (w,(R,G)) \geq n_0$.
By the pumping lemma we obtain the witness $j \leq i_1 < i_2 \leq j + n_0$.
Let $u = w_0 w_1 \ldots w_{i_1-1}$, $v= w_{i_1} w_{i_1+1} \ldots w_{i_2-1}$ and
$w' = w_{i_2} w_{i_2+1} \ldots$.
Since $w \in M$, by the pumping lemma for all $\ell \geq 0$ we have 
$u v^\ell w' \in M$.
This entails that all finite prefixes of the infinite word $uv^{\omega}$ are in 
$\pref{M}$. Since $M$ is closed, it follows that $uv^{\omega} \in M$.
Since $\mathrm{dist}_j(w,(R,G)) \geq n_0$ it follows that there is some 
request $i$ in position $j$, and there is no corresponding grant $i$ for 
the next $n_0$ steps.
Hence there is a position $j'$ in $v$ such that there is request $i$ at $j'$ 
and no corresponding grant in $v$, and thus
it follows that the word $u v^\omega \not\in \streett$.
This contradicts that $M \subseteq \streett$.
Hence it follows that $\fin(\streett) \subseteq \fstreettrg$.

	\item We now show the converse:
$\fin(\streett) \supseteq \fstreettrg$.
We have:
$$\begin{array}{rcl}
\fstreettrg & = & \displaystyle 
\{ w \mid \limsup_{k} \distk (w,(R,G)) < \infty \} \\
& = & \displaystyle \bigcup_{B \in \N} \{ w \mid \limsup_{k} \distk (w,(R,G)) \leq B \} \\
& = & \displaystyle \bigcup_{B \in \N} \bigcup_{n \in \N} \{ w \mid \forall k \geq n, \distk (w,(R,G)) \leq B \}
\end{array}$$
The language $\{ w \mid \forall k \geq n, \distk (w,(R,G)) \leq B \}$ is closed, $\omega$-regular, and included in 
$\streett$. 
Hence $\fstreettrg \subseteq \fin(\streett)$.
\end{enumerate}
The result follows. \hfill\qed
\end{proof}

Lemma~\ref{fins} naturally extends to finitary parity and finitary B\"uchi languages:

\begin{corollary}\label{finp}
The following assertions hold:
\begin{itemize}
\vspace{-2.5mm}
 \item For all $p : \Sigma \rightarrow \N$, we have $\fin(\parp) = \finp$;
 \item For all $F \subseteq \Sigma$, we have $\fin(\bucf) = \fbucf$.
\end{itemize}
\end{corollary}

B\"uchi languages are a special case of parity languages, and parity languages are in turn a special case of Streett languages. 
Since distance sequences for parity and B\"uchi languages have been defined as a special case of Streett languages, Corollary~\ref{finp} follows from Lemma~\ref{fins}.

The following lemma states that finitary B\"uchi languages are $\Sig{2}$-complete.

\begin{theorem}[Topological characterization of finitary languages]\label{thrm_topo}
The finitary B\"uchi, finitary parity and finitary Streett are $\Sig{2}$-complete.
\end{theorem}

\begin{proof}
We show that if $\emptyset \subset F \subset \Sigma$, then $\fbucf$ is $\Sig{2}$-complete. It follows from Corollary~\ref{finp} that $\fbucf \in \Sig{2}$.
We now show that $\fbucf$ is $\Sig{2}$-hard.
By Theorem~\ref{thrm_top} we have that $\bucf$ is $\Pi_2$-complete, hence
$\So \backslash \bucf$ is $\Sig{2}$-complete.
We present a topological reduction to show that $\So \backslash \bucf \preceq \fbucf)$.
Let $b : \So \rightarrow \So$ be the stuttering function defined as follows:
$$\begin{array}{ccccccc}
w & = & w_0 & w_1 & \ldots & w_n & \ldots \\
b(w) & = & w_0 & \underbrace{w_1 w_1}_2 & \ldots & \underbrace{w_n w_n \ldots w_n}_{2^n} & \ldots
\end{array}$$
The function $b$ is continuous.
We check that the following holds: 
$$\infi(w) \subseteq F \ \mbox{ iff } \ \exists B \in \N, 
\exists n \in \N, \forall k \geq n, \nextk (b(w),F) \leq B.$$
Left to right direction: assume that from the position $n$ of $w$, letters belong to $F$.
Then from the position $2^n - 1$, letters of $b(w)$ belong to $F$, then $\nextk (b(w),F) = 0$ for $k \geq 2^n - 1$.\\
Right to left direction: let $B$ and $n$ be integers such that for all $k \geq n$ we have $\nextk (b(w),F) \leq B$.
Assume $2^{k-1} > B$ and $k \geq n$, then the letter in position $2^k-1$ in $b(w)$ is repeated $2^{k-1}$ times,
thus $\nextk (b(w),F)$ is either $0$ or higher than $2^{k-1}$.
The latter is not possible since it must be less than $B$. 
It follows that the letter in position $k$ in $w$ belongs to $F$.
Hence we get $\So \backslash \bucf \preceq \fbucf$, so $\fbucf$ is $\Sig{2}$-complete.
From this we deduce the two other claims as special cases.\hfill\qed
\end{proof}

\section{Expressive power of finitary automata}

In this section we consider the finitary automata, and compare their expressive power to classical automata.
We then address the question of determinization.
Deterministic finitary automata enjoy nice properties that allows to describe languages they recognize using the $\fin$ operator.
As a by-product we get boolean closure properties of finitary automata.

\subsection{Comparison with classical automata}

Finitary conditions allow to express bounds requirements:

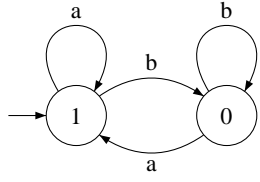
\begin{figure}
\centering
\begin{picture}(20,13)(0,0)
	\gasset{Nw=8,Nh=8}
 
	\node[Nmarks=i](1)(0,5){1}
	\node(0)(20,5){0}
	\drawedge[curvedepth=5](0,1){a}
	\drawedge[curvedepth=5](1,0){b}
	\drawloop[loopangle=90](0){b}
	\drawloop[loopangle=90](1){a}
\end{picture}
\caption{A finitary B\"uchi automaton $\A$}\label{auto1}
\end{figure}

\begin{example}[$\DFB \not\subseteq \omegareg$]\label{examp1}
Consider the finitary B\"uchi automaton shown 
in Fig.~\ref{auto1}, the state labeled~0 being its only final state.
Its language is $L_B = \set{ (b^{j_0} a^{f(0)}) \cdot 
(b^{j_1} a^{f(1)}) \cdot (b^{j_2} a^{f(2)}) \ldots \mid 
f : \N \rightarrow \N, f \mbox{ bounded, } \forall i \in \N, j_i \in \N }$.
Indeed, $0$-labeled state is visited while reading the letter $b$, and 
the $1$-labeled state is visited while reading the letter $a$. 
An infinite word is accepted iff the $0$-labeled state is 
visited infinitely often and there is a bound between two 
consecutive visits of the $0$-labeled state.
We can easily see that $L_B$ is not $\omega$-regular, using proof ideas from~\cite{BC06-LICS}: its complement would be $\omega$-regular, so it would contain ultimately periodic words, which is not the case.
\end{example}

However, finitary automata cannot distinguish between ``many b's'' and ``only b's'':

\begin{example}[$\DB \not\subseteq \NFB$]\label{examp2}
Consider the language of infinitely many $a$'s, \textit{i.e}, 
$L_I = \{ w \mid w \mbox{ has an infinite number of } a\}$.
The language $L_I$ is recognized by a simple deterministic B\"uchi automaton.
However, we can show that there is no finitary B\"uchi automata that recognizes $L_I$. Intuitively, such an automaton would, while reading the infinite word 
$w = ab\ ab^2\ ab^3\ ab^4\ldots ab^n\ldots \in L_I$, have to distinguish between all b's, otherwise it would accept a word with only b's at the end.
Assume towards contradiction that there exists $\A$ a non-deterministic finitary B\"uchi automaton with $N$ states recognizing $L_I$. Let us consider the infinite word $w$.
Since $w$ must be accepted by $\A$, there must be an accepting run $\rho$, represented as follows:
$$q_0 \xrightarrow{a} p_0 \xrightarrow{b} q_1 \ldots q_n \xrightarrow{a} p_n \xrightarrow{b^{n+1}} q_{n+1} \ldots$$
and
$$p_{n-1} \xrightarrow{b} q_{n,1} \xrightarrow{b} q_{n,2} \ldots \xrightarrow{b} q_{n,n-1} \xrightarrow{b} q_{n,n} = q_n \ldots$$
Since $\rho$ is accepting, there exists $B \in \N$, and $n \in \N$, such that for all $k \geq n$ we have $\distk (\rho,p) \leq B$. 
Let $c$ be the lowest priority infinitely visited in $\rho$. As $\rho$ is accepting, $c$ is even.
The state $p_{k-1}$ is in position $\frac{k\cdot (k+1)}{2}$ in $\rho$.
Let $k$ be an integer such that (a)~$\frac{k\cdot (k+1)}{2} \geq n$ and (b)~$k \geq (N+1)\cdot B$.
Let us consider the set of states $\{q_{k,1},\ldots, q_{k,k}\}$.
Since the distance function is bounded by $B$ from the $n$-th position, the priority $c$ appears at least once in each set of consecutively visited states of size $B$. 
Since $\frac{k\cdot (k+1)}{2} \geq n$ and $q_{k,1}$ is the state following $p_{k-1}$,
the latter holds from $q_{k,1}$.
Since $k \geq (N+1)\cdot B$, it appears at least $N+1$ times in $\{q_{k,1},\ldots,q_{k,k}\}$. Since there is $N$ states in $\A$, at least one state 
has been reached twice. We can thus iterate: the infinite word 
$w' = ab\ ab^2\ ab^3\ ab^4\ldots b^{k-1}a\ b^{\omega}$,
and the word $w'$ is accepted by $\A$. However, $w'\not \in L_I$ and hence we have 
a contradiction.
\end{example}

We summarize the results in the following theorem.

\begin{theorem} The following assertions hold:
(a)~$\DB \not\subseteq \NFB$;
(b)~$\DFB \not\subseteq \NB$.
\end{theorem}

\subsection{Deterministic finitary automata}
Given a deterministic complete automaton $\A$ with accepting condition $\Acc$, 
we will consider the language obtained by using $\fin(\Acc)$ as acceptance condition.
Treating the automaton as a transducer, 
we consider the following function: $C_{\A} : \So  \to Q^{\omega}$ 
which maps an infinite word $w$ to the unique run $\rho$ of $\A$ on $w$ (there is a unique run since $\A$ is deterministic and complete).
Then: $$\LA = \{w \mid C_{\A}(w) \in \Acc \} = C_{\A}^- (\Acc).$$
We will focus on the following property: $C_{\A}^- (\fin(\Acc)) = \fin(C_{\A}^- (\Acc))$, which follows from the following lemma. 
Deterministic complete automata, regarded as transducers, preserve topology and $\omega$-regularity. Hence applying the finitary operator $\fin$ to the input (the language $L$) or to the acceptance condition $\Acc$ is equivalent.

\begin{lemma}\label{lemm_det} 
For all $\A = (Q, \Sigma, q_0, \delta, \Acc)$ deterministic 
complete automaton, we have:
\begin{enumerate}
\vspace{-2.5mm}
\item for all $A \subseteq Q^{\omega}$, $A$ is closed 
$\Rightarrow  C_{\A}^- (A)$ closed ($C_{\A}$ is continuous).
\item for all $L \subseteq \So$, $L$ is closed 
$\Rightarrow  C_{\A} (L)$ closed ($C_{\A}$ is closed).
\item for all $A \subseteq Q^{\omega}$, $A$ is $\omega$-regular 
$\Rightarrow  C_{\A}^- (A)$ $\omega$-regular.
\item for all $L \subseteq \So$, $L$ is $\omega$-regular 
$\Rightarrow  C_{\A} (L)$ $\omega$-regular.
\end{enumerate}
\end{lemma}

\begin{proof}
We prove all the cases below.
\begin{enumerate}
	\item Let $A \subseteq Q^{\omega}$ such that $A$ is closed. 
Let $w$ be such that for all $n \in \N$ we have $w_0 \ldots w_n \in 
\pref{C_{\A}^- (A)}$. 
We define the run $\rho = C_{\A}(w)$ and show that $\rho = q_0 q_1 \ldots \in A$.
Since $A$ is closed, we will show for all $n \in \N$ we have 
$q_0 \ldots q_n \in \pref{A}$. 
From the hypothesis we have $w_0 \ldots w_{n-1} \in \pref{C_{\A}^- (A)}$, and 
then there exists an infinite word $u$ such that 
$C_{\A}(w_0 \ldots w_{n-1} u) \in A$. 
Let $C_{\A}(w_0 \ldots w_{n-1} u) = q_0 q'_1 \ldots q'_n \ldots$,
then we have $q_0 \xrightarrow{w_0} q'_1 \xrightarrow{w_1} q'_2
 \cdots \xrightarrow{w_{n-1}} q'_n \cdots$. 
Since $\A$ is deterministic, we get $q'_i = q_i$, and hence
$q_0 \ldots q_n \in \pref{A}$. 

	\item Let $L \subseteq \So$ such that $L$ is closed. 
Let $\rho = q_0 q_1 \ldots$ such that for all $n \in \N$ 
we have $q_0 \ldots q_n \in \pref{C_{\A}(L)}$. 
Then for all $n \in \N$, there exists a word $w_0 w_1 \ldots w_{n-1}$ such that
$q_0 \xrightarrow{w_0} q_1 \xrightarrow{w_1}q_2\dots\xrightarrow{w_{n-1}} q_n$,
and $w_0 w_1 \ldots w_{n-1} \in \pref{L}$.
We define by induction on $n$ an infinite nested sequence of finite words 
$w_0 w_1 \ldots w_n \in \pref{L}$.
We denote by $w$ the limit of this nested sequence of finite words.
We have that $\rho = C_{\A}(w)$. Since $L$ is closed, $w \in L$.

	\item Let $A \subseteq Q^{\omega}$ such that $A$ recognized by a B\"uchi automaton 
$\B = (Q_{\B}, Q, P_0, \tau, F)$. 
We define the B\"uchi automaton 
$\C = (Q \times Q_{\B}, \Sigma, \set{q_0} \times P_0, \gamma, Q_{\B} \times F)$,
where $(q_1,p_1) \xrightarrow{\sigma} (q_2,p_2)$ iff $q_1 \xrightarrow{\sigma} q_2$ in $\A$
and $p_1 \xrightarrow{q_1} p_2$ in $\B$.
We now show the correctness of our construction.
Let $w = w_0 w_1 \dots$ accepted by $\C$, then there exists an accepting run 
$\rho$, as follows:
\[
(q_0,p_0) \xrightarrow{w_0} (q_1,p_1) \xrightarrow{w_1} (q_2,p_2) \dots (q_n,p_n) \xrightarrow{w_n} (q_{n+1},p_{n+1}) \dots
\]
where the second component visits $F$ infinitely often. Hence:
$$(\dag) 
\left\lbrace
\begin{array}{cc}
q_0 \xrightarrow{w_0} q_1 \xrightarrow{w_1} q_2 \dots q_n \xrightarrow{w_n} q_{n+1} \dots \mbox{in } \A \\
p_0 \xrightarrow{q_0} p_1 \xrightarrow{q_1} p_2 \dots p_n \xrightarrow{q_n} p_{n+1} \dots \mbox{in } \B \quad
\end{array}
\right.$$
Hence from $(\dag)$, we have $C_{\A}(w) = q_0 q_1 \dots \in \La(\B) = A$, and
it follows that $w \in C_{\A}^- (A)$.
Conversely, let $w \in C_{\A}^- (A)$, then we 
have $\rho = C_{\A}(w) = q_0 q_1 \dots \in A = \La(\B)$.
Then the above statement $(\dag)$ holds, which entails that $w$ is 
accepted by $\C$.
It follows that $\C$ recognizes $C_{\A}^- (A)$.

	\item Let $L \subseteq \So$ such that $L$ is recognized by a B\"uchi automaton 
$\B = (Q_{\B}, \Sigma, P_0, \tau, F)$.
We define the B\"uchi automaton 
$\C = (Q \times Q_{\B}, Q, \set{q_0} \times P_0, \gamma, Q \times F)$,
where $(q,p_1) \xrightarrow{q} (q',p_2)$ iff 
there exists $\sigma \in \Sigma$, such that $q \xrightarrow{\sigma} q'$ in $\A$
and $p_1 \xrightarrow{\sigma} p_2$ in $\B$.
A proof similar to above show that $\C$ recognizes $C_{\A}(L)$.
\end{enumerate} 
The desired result follows.
\hfill\qed
\end{proof}

\begin{theorem}\label{findet}
For any deterministic complete automaton
$\A = (Q, \Sigma, q_0, \delta, \Acc)$ recognizing a language $L$, 
the finitary restriction of this automaton $\fin(\A) = (Q, \Sigma, q_0, \delta, \fin(\Acc))$ recognizes $\fin(L)$.
\end{theorem}

\begin{proof}
A word $w$ is accepted by $\fin(\A)$ iff $w \in C_{\A}^- (\fin(\Acc)) = \fin(C_{\A}^- (\Acc)) = \fin(L)$. \hfill\qed
\end{proof}

Theorem~\ref{findet} allows to extend all known results on deterministic classes to finitary deterministic classes: as a corollary, we have $\DFB \subset \DFP$ and $\DFP = \DFS$.

We now show that non-deterministic finitary parity automata are more 
expressive than deterministic finitary parity automata.
However, for every language $L \in \omegareg$ there exists $\A \in \DP$ such that $\A$ recognizes $L$, and by Theorem~\ref{findet} the deterministic finitary parity automaton $\fin(\A)$ recognizes $\fin(L)$.

\begin{corollary}
For every language $L \in \omegareg$ there is a deterministic finitary 
parity automata $\A$ such that $\LA=\fin(L)$. 
\end{corollary}

\begin{example}[$\DFP \subset \NFP$]\label{examp3}
As for Example~\ref{examp1} we consider the languages
$L_1 = \set{ (a^{j_0} b^{f(0)}) \cdot (a^{j_1} b^{f(1)}) \cdot (a^{j_2} b^{f(2)}) \ldots \mid f : \N \rightarrow \N, f \mbox{ bounded, } \forall i \in \N, j_i \in \N }$ and 
$L_2 = \set{ (a^{f(0)} b^{j_0}) \cdot (a^{f(1)} b^{j_1}) \cdot (a^{f(2)} b^{j_2}) \ldots \mid f : \N \rightarrow \N, f \mbox{ bounded, } \forall i \in \N, j_i \in \N }$.
It follows from Example~\ref{examp1} that both $L_1$ and $L_2$ belong to 
$\DFP$, hence to $\NFP$. A finitary parity automaton, relying on non-determinism, is easily built to recognize $L = L_1 \cup L_2$, hence $L \in \NFP$.
We can show that we cannot bypass this non-determinism, as by reading a word we have to decide well in advance which sequence will be bounded: a's or b's, \textit{i.e}, $L \notin \DFP$.
To prove it, we interleave words of the form $(a^* \cdot b^*)^* \cdot a^\omega$ and $(a^* \cdot b^*)^* \cdot b^\omega$, and use a pumping argument to reach a contradiction.
Assume towards contradiction that $L \in \DFP$, and let $\A$ be a deterministic complete finitary parity automaton with $N$ states that recognizes $L$.
Let $q_0$ be the starting state.
Since $a^\omega$ belongs to $L$, its unique run on $\A$ is accepting, and can be decomposed as follows:
$q_0 \xrightarrow{a^{n_0}} s_0 \xrightarrow{a^{p_0}} s_0 \xrightarrow{a^{p_0}} \dots$ 
where $s_0$ is the lowest priority visited infinitely often while reading $a^\omega$.
Then, $a^{n_0} b^\omega$ belongs to this $L$, its unique run on $\A$ is accepting, and has the following shape:
$q_0 \xrightarrow{a^{n_0}} s_0 \xrightarrow{b^{n'_0}} t_0 \xrightarrow{b^{p'_0}} t_0 \xrightarrow{b^{p'_0}} \dots$ 
where $t_0$ is the lowest priority visited infinitely often while reading $a^{n_0} b^\omega$.
Repeating this construction and by induction we have, as shown in Fig~\ref{fig_no_dfp}:
\begin{figure}
\begin{center}
\begin{picture}(105,25)(0,-4)
	\gasset{Nframe=n,Nw=6,Nh=5}
 
	\node(q_0)(0,5){$q_0$}
	\node(s_0)(15,15){$s_0$}
	\node(t_0)(30,5){$t_0$}
	\node(s_1)(45,15){$s_1$}

	\drawedge(q_0,s_0){$a^{n_0}$}

	\drawloop[loopangle=90](s_0){$a^{p_0}$}
	\drawedge(s_0,t_0){$b^{n'_0}$}
	\drawloop[loopangle=-90](t_0){$b^{p'_0}$}

	\drawedge(t_0,s_1){$a^{n_1}$}
	\drawloop[loopangle=90](s_1){$a^{p_1}$}

	\node(t_k-1)(75,5){$t_{k-1}$}
	\node(s_k)(90,15){$s_k$}
	\node(t_k)(105,5){$t_k$}

	\drawedge[dash={1.5}0](s_1,t_k-1){}

	\drawedge(t_k-1,s_k){$a^{n_k}$}
	\drawloop[loopangle=-90](t_k-1){$b^{p'_{k-1}}$}
	\drawloop[loopangle=90](s_k){$a^{p_k}$}
	\drawedge(s_k,t_k){$b^{n'_k}$}
	\drawloop[loopangle=-90](t_k){$b^{p'_k}$}
\end{picture}
\end{center}
\caption{Inductive construction showing that $L \notin \DFP$.}\label{fig_no_dfp}
\end{figure}
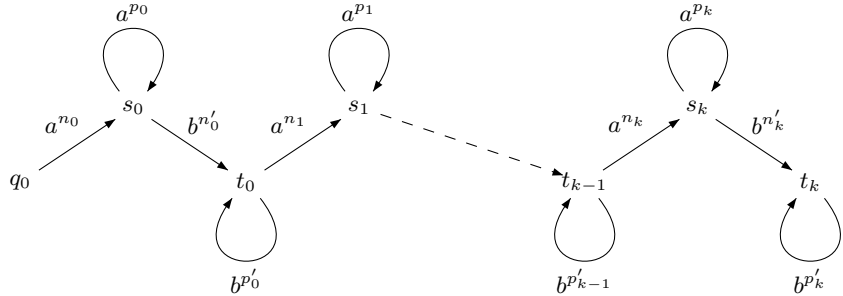
where $s_k$ is the lowest priority visited infinitely often while reading 
$a^{n_0} b^{n'_0} \dots a^{n_k} a^\omega$ and 
$t_k$ is the lowest priority visited infinitely often while reading 
$a^{n_0} b^{n'_0} \dots a^{n_k} b^{n'_k} b^\omega$.
There must be $i < j$, such that $t_i = t_j$. 
Let $u = a^{n_0} b^{n'_0} \dots b^{n'_i}$ and $v = b^{n'_{i+1}} \dots b^{n'_j}$, we have:
$$q_0 \xrightarrow{u} t_i \xrightarrow{a^{n_{i+1}}} s_{i+1} \xrightarrow{v} t_j = t_i$$
Consider the words $w = u \cdot (a^{n_{i+1}} \cdot v)^\omega$ and 
$$w^* = u \cdot (b^{p'_i} a^{n_i + p_i} v) \cdot (b^{2p'_i} a^{n_i + 2p_i} v) \dots (b^{kp'_i} a^{n_i + kp_i} v) \dots$$
$w$ must be accepted by $\A$ since it belongs to $L$. Hence $w^*$ is accepted as well, but does not belong to $L$. We have a contradiction, and the result follows.
\end{example}

\begin{theorem}\label{thrm_det_nondet} 
We have $\DFP \subset \NFP$.
\end{theorem}

Observe that Theorem~\ref{findet} does not hold for non-deterministic automata, since we have $\DP = \NP$ but $\DFP \subset \NFP$.

\subsection{Non-deterministic finitary automata}

We can show that non-deterministic finitary Streett automata can be 
reduced to non-deterministic finitary B\"uchi automata, and this 
would complete the picture of expressive power comparison.
We first show that non-deterministic finitary B\"uchi automata are 
closed under intersection, and use it to show Theorem~\ref{thrm_nsnb}.

\begin{lemma}\label{nfbconjunction}
$\NFB$ is closed under intersection.
\end{lemma}

\begin{proof}
Let $\A_1 = (Q_1,\Sigma,\delta_1,Q^1_0,F_1)$ and 
$\A_2 = (Q_2,\Sigma,\delta_2,Q^2_0,F_2)$ be two non-deterministic 
finitary B\"uchi automata.
Without loss of generality we assume both $\A_1$ and $\A_2$ to be complete. 
We will define a construction similar to the synchronous product construction,
where a switch between copies will happen while visiting $F_1$ or $F_2$.
The finitary B\"uchi automaton is
$\A = (Q_1 \times Q_2 \times \set{1,2},\Sigma,\delta,Q^1_0 \times Q^2_0 \times 
\set{1},F_1\times Q_2 \times\set{2} \cup Q_1 \times F_2 \times \set{1})$.
We define the transition relation $\delta$ below:
$$\begin{array}{rcl}
\delta & = & \set{((q_1,q_2,k),\sigma,(q'_1,q'_2,k)) \mid q'_1 \notin F_1, 
q_2' \notin F_2, (q_1,\sigma,q'_1) \in \delta_1, (q_2,\sigma,q'_2) \in \delta_2, 
k \in \set{1,2}} \\
& \cup & \set{((q_1,q_2,1),\sigma,(q'_1,q'_2,2)) \mid q'_1 \in F_1, (q_1,\sigma,q'_1) \in \delta_1, 
(q_2,\sigma,q'_2) \in \delta_2} \\
& \cup & \set{((q_1,q_2,2),\sigma,(q'_1,q'_2,1)) \mid q'_2 \in F_2, 
(q_1,\sigma,q'_1) \in \delta_1,
(q_2,\sigma,q'_2) \in \delta_2} 
\end{array}$$
Intuitively, the transition function $\delta$ is as follows:
the first component mimics the transition for automata $\A_1$, the 
second component mimics the transition for $\A_2$, and there is a switch 
for the third component from $1$ to $2$ visiting a state in $F_1$, 
and from $2$ to $1$ visiting a state in $F_2$.

We now prove the correctness of the construction.
Consider a word $w$ that is accepted by $\A_1$, and then 
there exists a bound  $B_1$ and a run $\rho_1$ in $\A_1$ such that eventually,
the number of steps between two visits to $F_1$ in $\rho_1$ is at most $B_1$;
and similarly,
there exists a bound  $B_2$ and a run $\rho_2$ in $\A_2$ such that eventually
the number of steps between two visits to $F_2$ in $\rho_2$ is at most $B_2$.
It follows that in our construction there is a run $\rho$ (that mimics
the runs $\rho_1$ and $\rho_2$) in $\A$ such that eventually within 
$\max\set{B_1, B_2}$ steps a state in 
$F_1 \times Q_2 \times \set{2} \cup Q_1 \times F_2 \times\set{1}$ is visited
in $\rho$. Hence $w$ is accepted by $\A$. 
Conversely, consider a word $w$ that is accepted by $\A$, and let $\rho$ be
a run and $B$ be the bound such that eventually between two visits to the accepting states
in $\rho$ is separated by at most $B$ steps.
Let $\rho_1$ and $\rho_2$ be the decomposition of the run $\rho$ in 
$\A_1$ and $\A_2$, respectively.
It follows that both in $\A_1$ and $\A_2$ the respective final states are 
eventually visited within at most $2\cdot B$ steps in $\rho_1$ and $\rho_2$,
respectively.
It follows that $w$ is accepted by both $\A_1$ and $\A_2$.
Hence we have $\LA = \La(\A_1) \cap \La(\A_2)$.
\hfill\qed
\end{proof}

\begin{theorem}\label{thrm_nsnb}
We have $\NFB = \NFP = \NFS$.
\end{theorem}

\begin{proof}
We will present a reduction of $\NFS$ to $\NFB$ and the result will follow.
Since the Streett condition is a finite conjunction of conditions 
$\infi(w) \cap R_i \neq \emptyset \Rightarrow \infi(w) \cap G_i 
\neq \emptyset$,
by Lemma~\ref{nfbconjunction} it suffices to handle the special case when 
$d = 1$.
Hence we consider a non-deterministic Streett automaton
$\A = (Q,\Sigma,\delta,Q_0,(R,G))$ with $(R,G) = (R_1,G_1)$.
Without loss of generality we assume $\A$ to be complete.
We construct a non-deterministic B\"uchi automaton $\A'=(Q\times\set{1,2,3}, 
\Sigma, \delta',Q_0 \times \set{1}, Q \times \set{2})$, where the transition relation 
$\delta'$ is given as follows:
\[
\begin{array}{rcl}
\delta' & = & 
\set{(q,1),\sigma,(q',j) \mid (q,\sigma, q') \in \delta, j \in \set{1,2} } \\
& \cup & 
\set{(q,2),\sigma,(q',2) \mid q' \notin R_1, (q,\sigma, q') \in \delta } \\
& \cup & \set{(q,2),\sigma,(q',3) \mid q' \in R_1, (q,\sigma, q') \in \delta} \\
& \cup & \set{(q,3),\sigma,(q',3) \mid q' \notin G_1, (q,\sigma,q') \in \delta } \\
& \cup & \set{(q,3),\sigma,(q',2) \mid q' \in G_1, (q,\sigma, q') \in \delta} 
\end{array}
\]
In other words, the state component mimics the transition of $\A$, and in the 
second component:
(a)~the automaton can choose to stay in component $1$, or switch to $2$; 
(b)~there is a switch from $2$ to $3$ upon visiting a state in $R_1$; and 
(b)~there is a switch from $3$ to $2$ upon visiting a state in $G_1$.
Consider a word $w$ accepted by $\A$ and an accepting run $\rho$ in $\A$, and 
let $B$ be the bound on the distance sequence.
We show that $w$ is accepted by $\A'$ by constructing an accepting run $\rho'$ in $\A'$.
We consider the following cases: 
\begin{enumerate}
	\item If infinitely many requests $R_1$ are visited in $\rho$, then in $\A'$
immediately switch to component $2$, and then mimic the run $\rho$ as a 
run $\rho'$ in $\A'$.
It follows that from some point $j$ on every request is granted within $B$ 
steps, and it follows that after position $j$, whenever the second component 
is $3$, it becomes~$2$ within $B$ steps. Hence $w$ is accepted by $\A$.
	\item If finitely many requests $R_1$ are visited in $\rho$, then after some 
point $j$, there are no more requests.
The automaton $\A'$ mimics the run $\rho$ by staying in the second component 
as~$1$ for $j$ steps, and then switches to component~$2$.
Then after $j$ steps we always have the second component as~$2$, and hence
the word is accepted. 
\end{enumerate}
Conversely, consider a word $w$ accepted by $\A'$ and consider the accepting
run $\rho'$. We mimic the run in $\A$.
To accept the word $w$, the run $\rho'$ must switch to the second component 
as~$2$, say after $j$ steps.
Then, from some point on whenever a state with second
component $3$ is visited, within some bound $B$ steps a state with
second component $2$ is visited.
Hence the run $\rho$ is accepting in $\A$.
Thus the languages of $\A$ and $\A'$ coincide, and the desired
result follows.
\hfill\qed
\end{proof}

Our results are summarized in Corollary~\ref{sum} and shown in Fig~\ref{fig_automata}.

\begin{corollary}\label{sum}
We have (a)~$\DFB \not\subseteq \omegareg$;
(b)~$\DFB \subset \DFP = \DFS \subset \NFB = \NFP = \NFS$;
(c)~$\DB \not\subseteq \NFB$;
(d)~$\omegareg \not\subseteq \NFB$.
\end{corollary}

\subsection{Closure properties}
 
\begin{theorem}[Closure properties]\label{closure} The following closure properties hold:
\begin{enumerate}
\vspace{-2.5mm}
\item $\DFP$ is closed under intersection.
\item $\DFP$ is not closed under union.
\item $\NFP$ is closed under union and intersection.
\item $\DFP$ and $\NFP$ are not closed under complementation.
\end{enumerate}
\end{theorem}

\begin{proof}
We prove all the cases below.
\begin{enumerate}
	\item Intersection closure for $\DFP$ follows from Theorem~\ref{findet} 
and from the observation that for all $L,L' \subseteq \So$ we have 
$\fin(L \cap L') = \fin(L) \cap \fin(L')$.
The observation is proved as follows. 
Let $M \in \Pi_1 \cap \omegareg$ and $M \subseteq L \cap L'$, 
then $M \subseteq \fin(L) \cap \fin(L')$, and
hence $\fin(L \cap L') \subseteq \fin(L) \cap \fin(L')$.
Conversely, let $M_1 \subseteq \fin(L)$ and $M_2 \subseteq \fin(L')$, then 
$M_1 \cap M_2 \in \Pi_1 \cap \omegareg$ and $M_1 \cap M_2 \subseteq L \cap L'$.
Hence $M_1 \cap M_2 \subseteq \fin(L \cap L')$, 
thus $\fin(L) \cap \fin(L') \subseteq \fin(L \cap L')$.

	\item Failure of closure under union for $\DFP$ follows from Example~\ref{examp3}.
	
	\item Union closure for $\NFP$ is easy and relies on non-determinism, while intersection closure follows from Lemma~\ref{nfbconjunction}, since $\NFP = \NFB$. 

	\item Failure of closure under complementation for $\DFP$ follows from items 1. and 2., since this closure together with intersection closure would imply union closure.
Failure of closure under complementation for $\NFP$ follows from Example~\ref{examp2}. Indeed, the language $L_F = \{a,b\}^{\omega}\ \backslash L_I = \{ w \mid w \mbox{ has a finite number of } a\}$ lies in $\NFP$; however, Example~\ref{examp2} shows that its complement is not expressible by non-deterministic finitary B\"uchi automata, hence nor by non-deterministic finitary parity automata.
\end{enumerate}
The result follows.
\hfill\qed
\end{proof}

\begin{figure}
\centering
\begin{tikzpicture}
   \draw[rotate=90] (0,0) ellipse (1cm and 2cm);
   \draw[rotate=90] (0,0) ellipse (1.5cm and 3cm);
   \draw[rotate=90] (0,0) ellipse (2cm and 4cm);

   \draw[rotate=60] (.4,1.8) ellipse (1cm and 1.5cm);

   \draw[rotate=90,color=gray,line width=2.5pt] (1.2,0) ellipse (1.5cm and 3.4cm);
   \draw (0,-0.7) node {$\DFB$}; 
   \draw (0,-1.2) node {$\DFP = \DFS$};
   \draw (0,-1.7) node {$\NFB = \NFP = \NFS$};  
   \draw (-2,2) node {$\DB$};
   \draw (0,2.4) node {$\omegareg$};  
\end{tikzpicture}
\caption{Expressive power classification}\label{fig_automata}
\end{figure}
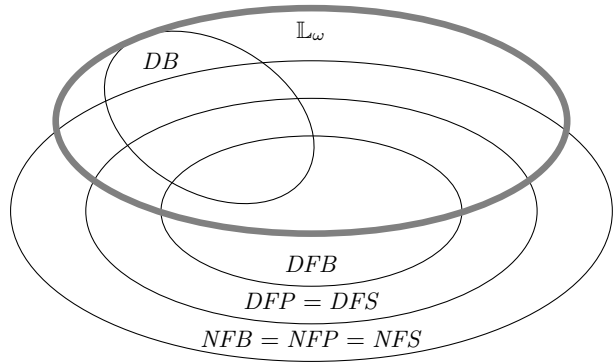

\section{Regular Expression Characterization}

In this section we address the question of giving a syntactical representation of finitary languages, using a special class of regular expressions.

The class of $\omega B$-regular expressions was introduced in the work 
of~\cite{BC06-LICS} as an extension of $\omega$-regular expressions, as an attempt to express bounds in regular languages.
To define $\omega B$-regular expressions, we need regular expressions and $\omega$-regular expressions.

Regular expressions define regular languages over finite words, and
have the following grammar:
$$L := \emptyset \mid \varepsilon \mid \sigma \mid L \cdot L \mid L^* \mid L + L; \quad \sigma \in \Sigma$$
In the above grammar, $\cdot$ stands for concatenation, $*$ for Kleene star and $+$ for union.
Then $\omega$-regular languages are finite union of $L \cdot L'^\omega$,
where  $L$ and $L'$ are regular languages of finite words.
The class of $\omega B$-regular languages, as defined in~\cite{BC06-LICS}, 
is described by finite union of $L \cdot M^\omega$, 
where $L$ is a regular language over finite words 
and $M$ is a $B$-regular language over infinite sequences of finite words.
The grammar for $B$-regular languages is as follows:
$$M := \emptyset \mid \varepsilon \mid \sigma \mid M \cdot M \mid M^* \mid M^B \mid M + M; \quad \sigma \in \Sigma$$
The semantics of regular languages over infinite sequences of finite words 
will assign to a $B$-regular expression $M$, a language in $(\Sigma^*)^\omega$.
The infinite sequence $\seq{u_0,u_1,\dots}$ will be denoted by $\vec{u}$.
The semantics is defined by structural induction as follows.
\begin{itemize}
\vspace{-2.5mm}
  \item $\emptyset$ is the empty language,
  \item $\varepsilon$ is the language containing the single sequence $(\varepsilon,\varepsilon,\dots)$,
  \item $a$ is the language containing the single sequence $(a,a,\dots)$,
  \item $M_1 \cdot M_2$ is the language $\set{\seq{u_0 \cdot v_0, u_1 \cdot v_1,\dots} \mid \vec{u} \in M_1, \vec{v} \in M_2}$,
  \item $M^*$ is the language $\set{\seq{u_1 \dots u_{f(1) - 1}, u_{f(1)} \dots u_{f(2) - 1}, \dots} \mid \vec{u} \in M, f : \N \to \N}$,
  \item $M^B$ is defined like $M^*$ but we additionally require the values $f(i+1) - f(i)$ to be bounded uniformly in $i$,
  \item $M_1 + M_2$ is $\set{\vec{w} \mid \vec{u} \in M_1 ,\vec{v} \in M_2, \forall i, w_i \in \set{u_i,v_i}}$.
\end{itemize}

\noindent Finally, the $\omega$-operator on sequences with nonempty words on infinitely many coordinates is: $\seq{u_0,u_1,\dots}^\omega = u_0 u_1 \dots$.
This operation is naturally extended to languages of sequences by taking the 
$\omega$ power of every sequence in the language.
The class of $\omega B$-regular languages is more expressive than $\NFB$, 
and this is due to the $*$-operator.
We will consider the following fragment of $\omega B$-regular languages where we do not use the $*$-operator for $B$-regular expressions (however, the $*$-operator is allowed for $L$, regular languages over finite words).
We call this fragment the star-free fragment of $\omega B$-regular languages.
In the following two lemmas we show that star-free $\omega B$-regular 
expressions express exactly $\NFB$.

\begin{lemma}\label{lemm_wb1}
All languages in $\NFB$ can be described by a star-free $\omega B$-regular expression.
\end{lemma}

\begin{proof}
Let $\A = (Q,\Sigma,\delta,Q_0,F)$ be a non-deterministic finitary 
B\"uchi automaton. 
Without loss of generality we assume $Q = \set{1,\dots,n}$.
Let $L_{q,q'} = \set{u \in \Sigstar \mid q \xrightarrow{u} q'}$ 
and $M^{\geq c}_q = \set{\vec{u} \mid (|u_i|)_i \mbox{ is bounded and } 
\forall i, q \xrightarrow{u_i} q}$.
Then $$\LA = \bigcup_{q_0 \in Q_0, q \in F} L_{q_0,q} \cdot (M_q)^\omega.$$ 
For all $q,q' \in Q$ we have $L_{q,q'} \subseteq \Sigstar$ is regular. 
We now show that for all $q \in Q$ the language $M_q$ is $B$-regular.
For all $0 \leq k \leq n$ and  $q,q' \in Q$, 
let $M^k_{q,q'} = \set{\vec{u} \mid (|u_i|)_i \mbox{ is bounded and } 
\forall i, q \xrightarrow{u_i} q' \mbox{ where all intermediate visited states are from } 
\set{1,\dots,k}}$.
We show by induction on $0 \leq k \leq n$ that for all $q,q' \in Q$ 
the language $M^k_{q,q'}$ is $B$-regular.
The base case $k = 0$ follows from observation:
$$M^0_{q,q'} = \left \{ \begin{array}{ll}
a_1 + a_2 + \dots + a_l & \mbox{if } q \neq q' \mbox{ and } (q,a,q') \in \delta \iff \exists i \in \set{1,\dots,l}, a = a_i \\
\varepsilon + a_1 + a_2 + \dots + a_l & \mbox{if } q = q' \mbox{ and } (q,a,q') \in \delta \iff \exists i \in \set{1,\dots,l}, a = a_i \\
\emptyset & \mbox{otherwise}
\end{array} \right.$$
The inductive case for $k > 0$ follows from observation:
$$M^k_{q,q'} = M^{k-1}_{q,k} \cdot (M^{k-1}_{k,k})^B \cdot M^{k-1}_{k,q'} + M^{k-1}_{q,q'}$$
Since $M^n_{q,q} = M_{q}$, we conclude that $\LA$ is described by a star-free $\omega B$-regular expression.
\hfill\qed
\end{proof}

\begin{lemma}\label{lemm_wb2}
All languages described by a star-free $\omega B$-regular expression 
is recognized by a non-deterministic finitary B\"uchi automaton.
\end{lemma}

\begin{proof}
To prove this result, we will describe automata reading infinite sequences of 
finite words, and corresponding acceptance conditions.
Let $\A = (Q,\Sigma,\delta,Q_0,F)$ a finitary B\"uchi automaton. While reading an infinite sequence $\vec{u}$ of finite words,
$\A$ will accept if the following conditions are satisfied:
(1)~$\exists q_0 \in Q_0, q_1,q_2,\ldots \in F, \forall i \in \N$, we have $q_i \xrightarrow{u_i} q_{i+1}$
and (2)~$(|u_n|)_n$ is bounded.

We show that for all $M$ star-free $B$-regular expression, 
there exists a non-deterministic finitary B\"uchi automaton accepting $M^B$,
language of infinite sequence of finite words, as described above.
We proceed by induction on $M$.
\begin{itemize}
 \item The cases $\emptyset, \varepsilon$ and $a \in \Sigma$ are easy.
 \item From $M$ to $M^B$, the same automaton for $M$ works for $M^B$ as well, since $B$ is idempotent.
 \item From $M_1,M_2$ to $M_1 + M_2$: this involves non-determinism. 
The automaton guesses for each finite word which word is used.
Let $\A_1 = (Q_1,\Sigma, \delta_1,Q_1^0,F_1)$ and $\A_2 = (Q_2,\Sigma, \delta_2,Q_2^0,F_2)$
two non-deterministic finitary B\"uchi automata accepting $M_1^B$ and $M_2^B$, respectively.
For $k \in \set{1,2}$ and $T \subseteq Q_k$, we define 
$\Final(T) = \set{q' \in F_k \mid \exists q \in T, \exists u \in \Sigstar, q \xrightarrow{u}_{\A_k} q'}$
to be the state of final states reachable from a state in $T$. 
We denote by $\Final^k$ the $k$-th iteration of $\Final$, e.g., 
$\Final^3(T)=\Final(\Final(\Final(T)))$.

We define a finitary B\"uchi automaton:
$$\A = (\underbrace{(Q_1 \times 2^{Q_1}) \cup (Q_2 \times 2^{Q_1})}_{\mbox{computation states}} \ 
\cup \ \underbrace{2^{Q_1} \times 2^{Q_2} }_{\mbox{guess states}}, \Sigma, \delta, (Q_1^0,Q_2^0), F)$$
where 
$$\begin{array}{rclr}
\delta & = & \set{((Q,Q'),\varepsilon,(q,\Final(Q'))) \mid q \in Q} & 
\quad (\mbox{guess is } 1) \\
& \cup & \set{((Q,Q'),\varepsilon,(q',\Final(Q))) \mid q' \in Q'} & 
\quad (\mbox{guess is } 2) \\
& \cup & \set{((q,T),\sigma,(q',T)) \mid (q,\sigma,q') \in \delta_1 \cup \delta_2} \\
& \cup & \set{((q_1,T),\varepsilon,(\set{q_1},T)) \mid q_1 \in F_1} \\
& \cup & \set{((q_2,T),\varepsilon,(T,\set{q_2})) \mid q_2 \in F_2} \\
\end{array}$$
There are two kinds of states.
Computation states are $(q,T)$ where $q \in Q_1$ and $T \subseteq Q_2$ 
(or symmetrically $q \in Q_2$ and $T \subseteq Q_1$), where $q$ is the current 
state of the automaton that has been decided to use for the current finite 
word, and $T$ is the set of final states of the other automaton that would 
have been reachable if one had chosen this automaton.
Guess states are $(Q,Q')$, 
where $Q$ is the set of states from $\A_1$ one can start reading the next word,
and similarly for $Q'$.

We now prove the correctness of our construction.
Consider an infinite sequence $\vec{w}$ accepted by $\A$, and consider 
an accepting run $\rho$.
There are three cases:
\begin{enumerate}
 \item either all guesses are $1$;
 \item or all guesses are $2$;
 \item else, both guesses happen.
\end{enumerate}
The first two cases are symmetric. In the first, we can easily see that 
$\vec{w}$ is accepted by $\A_1$, and similarly in the second
$\vec{w}$ is accepted by $\A_2$.

We now consider the third case.
There are two symmetric subcases: either the first guess is $1$, then 
\[
\rho = (Q_1^0, Q_2^0) \cdot (q^0_1,\Final(Q_2^0)) \dots,
\]
with $q^0_1 \in Q_1^0$; or symmetrically the first guess is $2$, then
\[
\rho = (Q_1^0,Q_2^0) \cdot (q^0_2,\Final(Q_1^0)) \dots,
\]
with $q^0_2 \in Q_2^0$.
We consider only the first subcase. Then 
\[
\rho = (Q_1^0,Q_2^0) \cdot (q^0_1,\Final(Q_2^0)) \dots (q^1_1,\Final(Q_2^0)) \cdot (\set{q^1_1},\Final(Q_2^0)) \dots,
\]
where $u_0$ is a finite prefix of $\vec{w}^\omega$ such that 
$q^0_1 \xrightarrow{u_0} q^1_1$ in $\A_1$ and $q^1_1 \in F_1$.
We denote by $\rho_0$ the finite prefix of $\rho$ up to $(q^1_1,\Final(Q_2^0))$.
Let $k$ be the first time when guess is $2$: then 
\[
\rho = \rho_0 \cdot \rho_1 \cdot \rho_{k-1} \cdot (\set{q^k_1},\Final^k(Q_2^0)) \cdot (q^0_2,\Final(\set{q_k})) \dots,
\]
where $q^0_2 \in \Final^k(Q_2^0)$ and for $1 \leq i \leq k-1$, we have
\[
\rho_i = (\set{q^i_1},\Final^i(Q_2^0)) \cdot (q^i_1,\Final^{i+1}(Q_2^0)) \dots (q^{i+1}_1,\Final^{i+1}(Q_2^0)),
\]
and $u_i$ is a finite word such that $q^i_1 \xrightarrow{u_i} q^{i+1}_1$ in $\A_1, q^{i+1}_1 \in F_1$ and
$u_0 u_1 \dots u_{k-1}$ finite prefix of $\vec{w}^\omega$.
Since $q^0_2 \in \Final^k(Q_2^0)$, there exists $v_0,v_1,\dots,v_{k-1}$ finite words 
and $q^1_2,\dots,q^k_2 \in F_2$ such that:
$q^0_2 \xrightarrow{v_0} q^1_2 \xrightarrow{v_1} \dots \xrightarrow{v_{k-1}} q^k_2$.
Then we can repeat this by induction, constructing $\vec{u} \in M_1^B$ and $\vec{v} \in M_2^B$, such that
for all $i \in \N$, we have $w_i \in \set{u_i,v_i}$.

Conversely, let $\vec{u} \in M_1^B$ and $\vec{v} \in M_2^B$, and $\vec{w}$ such that
$\forall i \in \N, w_i \in \set{u_i,v_i}$. Using $\A_1$ when $w_i = u_i$ and $\A_2$ otherwise,
one can construct an accepting run for $\vec{w}$ and $\A$.
Hence $\A$ recognizes $(M_1 + M_2)^B$.

 \item From $M_1,M_2$ to $M_1 \cdot M_2$: 
the automaton keeps tracks of pending states while reading the other word.
Let $\A_1 = (Q_1,\Sigma, \delta_1,Q_1^0,F_1)$ and $\A_2 = (Q_2,\Sigma, \delta_2,Q_2^0,F_2)$
two non-deterministic finitary B\"uchi automata accepting $M_1^B$ and $M_2^B$, respectively.
Let $\A = ((Q_1 \times F_2) \cup (Q_2 \times F_1), \Sigma, \delta, Q_1^0 \times Q_2^0, F_1 \times F_2)$, 
where $$\begin{array}{rcl}
\delta & = & \set{((q,f),\sigma,(q',f)) \mid (q,\sigma,q') \in \delta_1, f \in F_2} \\
& \cup & \set{((q,f),\sigma,(q',f)) \mid (q,\sigma,q') \in \delta_2, f \in F_1} \\
& \cup & \set{((q_1,f),\varepsilon,(f,q_1)) \mid q_1 \in F_1} \\
& \cup & \set{((q_2,f),\varepsilon,(f,q_2)) \mid q_2 \in F_2} \\
\end{array}$$
Intuitively, the transition relation is as follows: either one is reading using $\A_1$ or $\A_2$.
In both cases, the automaton remembers the last final state visited while 
reading in the other automaton in order to restore this state for the next word.
Let $\vec{w}$ accepted by $\A$, an accepting run is as follows:
\[
(q_1^0,q_2^0) \xrightarrow{w_0} (q_1^1,q_2^1) \xrightarrow{w_1} \ldots (q_1^i,q_2^i) \xrightarrow{w_i} (q_1^{i+1},q_2^{i+1}) \ldots
\]
where $(q_1^0,q_2^0) \in Q_1^0 \times Q_2^0$, for all $i \geq 1$, we have
$(q_1^i,q_2^i) \in F_1 \times F_2$ and $(|w_n|)_n$ bounded.
From the construction, for all $i \in \N$, we have 
$w_i = u_i^0 \cdot v_i^0 \cdot u_i^1 \cdot v_i^1 \ldots u_i^{k_i} \cdot v_i^{k_i}$,
where 
\[
q_1^i = q_1^i(0) \xrightarrow{u_i^0} q_1^i(1) \xrightarrow{u_i^1} q_1^i(2) \ldots \xrightarrow{u_i^{k_i}} q_1^i(k_i +1) = q_1^{i+1} \quad \mbox{ in } \A_1
\]
\[
q_2^i = q_2^i(0) \xrightarrow{v_i^0} q_2^i(1) \xrightarrow{v_i^1} q_2^i(2) \ldots \xrightarrow{v_i^{k_i}} q_2^i(k_i +1) = q_2^{i+1} \quad \mbox{ in } \A_2
\]
the states $(q_1^i(k),q_2^i(k))$ belong to $F_1 \times F_2$.
We define $u_i = u_i^0 u_i^1 \ldots u_i^{k_i}$ and $v_i = v_i^0 v_i^1 \ldots v_i^{k_i}$.
From the above follows that $\vec{u}$ and $\vec{v}$ are accepted by $\A_1$ and $\A_2$, respectively.
Then $\vec{w} \in (M_1 \cdot M_2)^B$.

Conversely, a sequence in $(M_1 \cdot M_2)^B$ is clearly accepted by $\A$.
Hence $\A$ recognizes $(M_1 \cdot M_2)^B$.
\end{itemize}

We now prove that all star-free $\omega B$-regular expressions are recognized 
by a non-deterministic finitary B\"uchi automaton.
Since $\NFB$ are closed under finite union (Theorem~\ref{closure}), we
only need to consider expressions $L \cdot M^\omega$,
where $L \subseteq \Sigstar$ is regular language of finite words and 
$M$ star-free $B$-regular expression.
The constructions above ensure that there exists 
$\A_M = (Q_M,\Sigma,\delta_M,Q_M^0,F_M)$, a 
non-deterministic finitary B\"uchi automaton that recognizes the 
language $M^B$ of infinite sequences.
Let $\A_L = (Q_L,\Sigma,\delta_L,Q_L^0,F_L)$ be a finite automaton over 
finite words that recognizes  $L$.
We construct a non-deterministic finitary B\"uchi automaton as follows:
$\A = (Q_L \cup Q_M,\Sigma,\delta,Q^0_L,F_M)$ where 
$\delta = \delta_L \cup \delta_M \cup \set{(q,\varepsilon,q') \mid q \in F_L, q' \in Q_M^0}$.
In other words, first $\A$ simulates $\A_L$, and when a finite prefix is 
recognized by $\A_L$, then $\A$ turns to $\A_M$ and simulates it.

We argue that $\A$ recognizes $L \cdot M^\omega$.
Let $w$ accepted by $\A$, and $u$ the finite prefix read by $\A_L$, $w = u\cdot v$.
From $v$ infinite word, we define $\vec{v}$ an infinite sequence of finite 
words by sequencing $v$ each time a final state (i.e., from $F_L$) is visited. 
The sequence $\vec{v}$ is accepted by $\A_M$, hence belongs to $M^B$, and since
$\vec{v}^\omega = v$, we have $v \in (M^B)^\omega = M^\omega$, and finally 
$w \in L \cdot M^\omega$.
Conversely, let $w = u \cdot \vec{v}^\omega$, where $u \in L$ and $\vec{v} \in M^B$.
Let $q_0 \in Q_L^0, q \in F_L$ such that $q_0 \xrightarrow{u} q$.
Let $q' \in Q_0, q_1,q_2,\ldots \in F_L$, such that for all $i \in \N$ 
we have $q_i \xrightarrow{v_i} q_{i+1}$.
The \emph{key}, yet simple observation is that for all star-free 
$B$-regular expressions $M$ and for all $\vec{v} \in M$ we have 
$(|v_n|)_n$ is bounded.
This is straightforward by induction on $M$.
Hence, from position $|u|$, the set $F_L$ is visited infinitely many times,
and there is a bound between two consecutive visits. Thus $w$ is accepted by $\A$.
\hfill\qed
\end{proof}

The following theorem follows from Lemma~\ref{lemm_wb1} and 
Lemma~\ref{lemm_wb2}.

\begin{theorem}\label{starfree}
$\NFB$ has exactly the same expressive power as star-free $\omega B$-regular 
expressions.
\end{theorem}

\section{Decision Problems}
In this section we consider the complexity of the decision 
problems for finitary languages.
We present the results for finitary B\"uchi automata for 
simplicity, but the arguments for finitary parity and 
Streett automata are similar.

For the proofs of the results of this section we need to consider 
co-B\"uchi conditions (dual of B\"uchi conditions): given a set $F$,
it requires that elements that appear infinitely often are 
outside $F$, in other words, elements in $F$ appear only finitely often.
It maybe noted that co-B\"uchi and finitary co-B\"uchi conditions coincide.
We will also consider co-finitary B\"uchi condition, that is the 
complement of a finitary B\"uchi condition: given a set $F$ co-finitary B\"uchi condition for $F$ is the complement of $\fbucf$, that is $\So \backslash \fbucf$.

\begin{lemma}\label{lem-decision}
Let $\A = (Q, \Sigma, Q_0, \delta_, F_b, F_c)$ be an automaton with 
$F_b$ and $F_c$ are subsets of $Q$. 
Consider the acceptance condition $\Phi_1$ as the conjunction of 
the finitary B\"uchi condition with set $F_b$, and the co-finitary B\"uchi condition 
with set $F_c$; and the acceptance condition $\Phi_2$ as the 
conjunction of B\"uchi condition with set $F_b$, and the co-B\"uchi condition 
with set $F_c$.
The following assertions hold:
\begin{enumerate}
\item The answer of the emptiness problem of $\A$ for $\Phi_2$ is Yes iff there is a 
cycle $C$ in $\A$ such that $C \cap F_b \neq \emptyset$ and $C \cap F_c = \emptyset$.

\item The answer of the emptiness problem for $\Phi_1$ and $\Phi_2$ coincide.

\item The emptiness problem for $\Phi_1$ is decidable in $\NLOGSPACE$.

\end{enumerate}
\end{lemma}

\begin{proof}
We prove the results as follows. 
\begin{enumerate}
\item We first prove parts 1. and 2. 
Without loss of generality we assume that for all $q \in Q$, there exists a path from an initial state $q_0 \in Q$ to $q$ (otherwise we can delete $q$).
If there is a cycle $C$ with $C \cap F_b \neq \emptyset$ and $C \cap F_c = \emptyset$, then consider a finite word $u$ to reach $C$, and a word $v$ that execute $C$.
The word $u \cdot v^\omega$ is a witness that $\A$ with $\Phi_1$ as well as $\Phi_2$ is non-empty.
Conversely, the condition $\Phi_2$ is a Rabin 1-pair condition, and by existence 
of memoryless strategies for Rabin condition~\cite{EJ88-FOCS}, it follows that if 
$\A$ is non-empty for $\Phi_2$, then there must be a cycle $C$ in $\A$ such that $C \cap F_b \neq \emptyset$ and $C \cap F_c = \emptyset$.
The condition $\Phi_1$ can be specified as a finitary parity condition 
with three priorities ($1,2,3$) by assigning priority~1 to states in $F_c$, 
$2$ to states in $F_b \setminus F_c$, and $3$ to the rest. 
By existence of memoryless strategies for finitary parity objectives~\cite{CHH09-ToCL},
it follows that if $\A$ is non-empty for $\Phi_1$, then there must be a cycle $C$ in $\A$ such that $C \cap F_b \neq \emptyset$ and $C \cap F_c = \emptyset$.
The result follows.

\item The result follows from the emptiness problem of non-deterministic Rabin 1-pair automata. 
The basic idea of the proof is as follows: we show that the witness cycle $C$ 
can be guessed and verified in logarithmic space.
The guesses are as follows: (a)~ first the initial prefix of the path to $C$ is guessed
by guessing one state (the next state) at a time 
(hence only one guess is made at a time which is logarithmic space), 
(b)~then the starting state of the cycle $C$ is guessed and stored (again in 
logarithmic space), and (c)~the cycle is guessed by again considering one 
state at a time and at each step it is verified that the state generated is in $Q \setminus F_c$; 
(d)~one state in the cycle such that the state is in $F_b$ is guessed and 
verified; and 
(e)~finally it is checked that the cycle is completed by visiting the starting
state of the cycle.
Hence at every step only constantly many guesses are made, stored and verified. 
The $\NLOGSPACE$ upper bound follows.
\end{enumerate}
The desired result follows.
\hfill\qed
\end{proof}

\begin{theorem}[Decision problems]\label{thrm_decision}
The following assertions hold:
\begin{enumerate}
\vspace{-2.5mm}
	\item \emph{(Emptiness).} Given a finitary B\"uchi automaton $\A$, whether 
$\LA = \emptyset$ is $\NLOGSPACE$-complete and can be decided
in linear time.

	\item \emph{(Universality).} Given a finitary B\"uchi automaton $\A$
whether $\LA=\Sigma^\omega$ is $\PSPACE$-complete.
 
	\item \emph{(Language inclusion).} 
Given two finitary B\"uchi automata $\A$ and $\B$, whether $\LA \subseteq \LB$
is $\PSPACE$-complete.
\end{enumerate}
\end{theorem}

\begin{proof}
We show the three parts of the proof.
\begin{enumerate}
\item The $\NLOGSPACE$ upper bound follows from Lemma~\ref{lem-decision}: 
we consider the special case where the set $F_c$ is empty.
The $\NLOGSPACE$ lower bound follows from $\NLOGSPACE$-hardness of reachability problem in a 
directed graph: given $s$ and $t$ two vertices, is there a path from $s$ to $t$? Given a directed graph and $s,t$ two vertices, the corresponding automaton has $s$ as initial vertex, $t$ as unique final vertex, and we add a self-loop over $t$. Then there is a path from $s$ to $t$ if and only if the language accepted by this finitary B\"uchi automaton is non-empty.
This concludes since co-$\NLOGSPACE = \NLOGSPACE$.

\item The $\PSPACE$ upper bound will follow from the following $\PSPACE$ upper bound 
for language inclusion, item 3. 
The $\PSPACE$ lower bound follows from the $\PSPACE$ lower bound for finite automata.
The universality problem for automata over finite words is $\PSPACE$-hard even when all the 
accepting states are absorbing~\cite{MeySto72}.
For such automata over finite words the acceptance is the same as for finitary B\"uchi condition. 
The result follows.

\item The $\PSPACE$ lower bound follows from item 2. by the $\PSPACE$-hardness for universality.
We now present the $\PSPACE$ upper bound.
Let $\A = (Q_A, \Sigma, Q_{A,0}, \delta_A, F_A)$ and $\B = (Q_B, \Sigma, Q_{B,0}, \delta_B, F_B)$ be two 
finitary B\"uchi automata. 
Let $\A \times \overline{\B} = (Q_A \times 2^{Q_B}, \Sigma, (Q_{A,0},Q_{B,0}), \delta, F_b, F_c)$ be an automaton where
for all $s \in Q_A, S \subseteq Q_B$ and $\sigma \in \Sigma$, 
$$\delta((s,S),\sigma) = \bigcup_{q \in S} \set{(s',q') \mid s' \in \delta_A(s,\sigma), q' \in \delta_B(q,\sigma)}$$
and $F_b = \set{(s,S) \mid s \in F_A}$ and $F_c = \set{(s,S) \mid S \cap F_B = \emptyset}$.
In other words $\A \times \overline{\B}$ is synchronous product of $\A$ and the power set (subset construction) of $\B$.
The acceptance condition is the conjunction of the finitary B\"uchi condition
with set $F_b$ and co-finitary B\"uchi condition with set $F_c$. 

We claim that $\La(\A \times \overline{\B}) = \emptyset$ iff $\LA \subseteq \LB$.

Assume $\La(\A \times \overline{\B}) \neq \emptyset$, then 
there is a cycle $C$ such that $C \cap F_b \neq \emptyset$ and $C \cap F_c = \emptyset$.
The lasso word that executes the finite path to reach $C$ and then execute it
forever is a witness word that is accepted by $\A$ but not by $\B$. 
Hence $\LA \not\subseteq \LB$.

Assume $\La(\A \times \overline{\B}) = \emptyset$, then
every words accepted by $\A$ is not accepted by $\overline{\B}$, hence accepted by $\B$.
Thus $\LA \subseteq \LB$.

Since the construction is exponential and the non-emptiness problem can be 
decided in $\NLOGSPACE$ (Lemma~\ref{lem-decision}), 
we obtain a $\NPSPACE = \PSPACE$ upper bound.
\end{enumerate}
The result follows.
\hfill\qed
\end{proof}

\medskip\noindent{\bf Acknowledgements.} We thank Thomas Colcombet for explaining to us results related to $\omega B$-automata.

\bibliography{people,short,papers}
\bibliographystyle{alpha}

\end{document}